%% file: main.tex
\documentclass[runningheads]{llncs}
\usepackage[frozencache,cachedir=.]{minted}
\PassOptionsToPackage{usenames}{xcolor}
\PassOptionsToPackage{dvipsnames}{xcolor}

\usepackage[utf8]{inputenc}
\usepackage[english]{babel}
\usepackage{amsmath,amsthm}
\usepackage{amssymb}
\usepackage{array}
\usepackage{graphicx}
\usepackage{float}
\usepackage{multirow}
\usepackage{hyperref}
\usepackage{mathtools}
\usepackage{caption}
\usepackage{subcaption}
\usepackage{tikz}
\usepackage{minted}
\usepackage{url}

\usepackage{breakurl}
\usepackage[noend]{algorithm2e}
\SetKwIF{If}{ElseIf}{Else}{if}{:}{else if}{else :}{end if}%
\SetKwFor{For}{for}{:}{end for}%

\usepackage{boxedminipage}
\usepackage{color}
\newcommand{\contract}{\textcolor{entry}{\texttt{contract}}}

\newcommand{\oninput}{\textcolor{entry}{\texttt{function}}}
\newcommand{\view}{\textcolor{entry}{\texttt{view}}}

\newcommand{\public}{\textcolor{entry}{\texttt{public}}}

\newcommand{\memory}{\textcolor{entry}{\texttt{memory}}}

\newcommand{\uint}{\textcolor{entry}{\texttt{uint}}}

\newcommand{\struct}{\textcolor{entry}{\texttt{struct}}}
\newcommand{\returns}{\textcolor{entry}{\texttt{returns}}}
\newcommand{\return}{\textcolor{entry}{\texttt{return}}}
\newcommand{\trues}{\textcolor{entry}{\texttt{true}}}
\newcommand{\falses}{\textcolor{entry}{\texttt{false}}}
\newcommand{\is}{\textcolor{entry}{\texttt{is}}}
\newcommand{\new}{\textcolor{entry}{\texttt{new}}}

\newcommand{\comment}[1]{\texttt{\textcolor{OliveGreen}{#1}}}
\definecolor{burntorange}{rgb}{0.8, 0.33, 0.0}
\newcommand{\requirecomment}[1]{\texttt{\textcolor{burntorange}{#1}}}

\newcommand{\bool}{\textcolor{entry}{\texttt{bool}}}

\newcommand{\require}{\textcolor{entry}{\texttt{require}}}

\newcommand{\ifs}{\textcolor{entry}{\texttt{if}}}

\newcommand{\for}{\textcolor{entry}{\texttt{for}}}

\colorlet{iomsg}{MidnightBlue}
\colorlet{party}{brown}
\colorlet{entry}{NavyBlue}
\colorlet{string}{BlueViolet}
\SetKwComment{Comment}{$\#$\ }{}
\SetKwRepeat{Do}{do}{while}

\hypersetup{
colorlinks=true,
linkcolor=black,
filecolor=magenta,      
urlcolor=black,
}

\newif\ifsubmission\submissiontrue

\begin{document}

\title{Sliding Window Challenge Process for Congestion Detection}
\authorrunning{A. Lotem et al.}
\author{Ayelet Lotem$^1$, Sarah Azouvi$^2$, Patrick McCorry$^{3}$ \and Aviv Zohar$^1$}
\institute{$^1$The Hebrew University of Jerusalem, \email{\{ayelem02,avivz\}@cs.huji.ac.il} \\
$^2$Protocol Labs, \email{sarah.azouvi@protocol.ai} \\
$^3$Infura, \email{stonecoldpat@gmail.com}}

\maketitle         

\input{macros}
\begin{abstract}

	Many prominent smart contract applications such as payment channels, auctions, and voting systems often involve a mechanism in which some party must respond to a challenge or appeal some action within a fixed time limit. This pattern of challenge-response mechanisms poses great risks if, during periods of high transaction volume, the network becomes congested. In this case, fee market competition can prevent the inclusion of the response in blocks, causing great harm. As a result, responders are allowed long periods to submit their response and overpay in fees.
	To overcome these problems and improve challenge-response protocols, we suggest a secure mechanism that detects congestion in blocks and adjusts the deadline of the response accordingly. 
	The responder is thus guaranteed a deadline extension should congestion arise.  We lay theoretical foundations for congestion signals in blockchains and then proceed to analyze and discuss possible attacks on the mechanism and evaluate its robustness.
	Our results show that in Ethereum, using short response deadlines as low as 3 hours, the protocol has $>99\%$ defense rate from attacks even by miners with up to $33\%$ of the computational power. 
	Using shorter deadlines such as one hour is also possible with a similar defense rate for attackers with up to $27\%$ of the power.
	
\end{abstract}

\section{Introduction}
DeFi platforms constructed over blockchains such as Ethereum have seen a recent boom of activity and interest. Their growing ecosystem allows for increasingly complex financial interactions executed in a fully decentralized manner. The main building blocks used to construct these platforms are the smart contracts that define the rules of interaction in code.

Smart contracts enable a wide range of applications, such as auctions, voting systems, and second layer protocols (e.g., payment channels) that operate above the blockchain layer. They typically provide rules that allow them to act as an automated adjudicator in case conflicts between participants arise.

For many applications, interactions with smart contracts are time dependent and are even subject to deadlines, meaning that in some cases, transactions added after a specific moment will effectively be rejected. For example in the case of auctions, a bid must be received before the end of the auction otherwise it is not valid.
Another example appears in the context of payment channels~\cite{gudgeon2020sok} where participants have a limited interval of time to dispute the division of funds if they disagree with their peers.

A major weakness of such deadlines is that in cases where the blockchain is congested, users that submit transactions will not have them included in blocks in time. In fact, several attacks and failures can be attributed directly to this weakness (we provide some examples below). One mitigation often employed by participants is to offer higher fees for transactions with deadlines which means users are usually overpaying. Another is to extend the deadlines which greatly delays processing and settlement within the context of the relevant smart contract. In many cases, transaction fees and deadlines are decided upon in advance, before the exact conditions that will prevail when the transaction is actually transmitted are known, which causes participants to take wider safety margins and increases costs further.
Due to the well-known scalability issues of blockchains~\cite{bano2019sok}, we expect congested periods to become increasingly more common, which will directly impact the design of time-sensitive smart contracts.

\vspace{2pt} \noindent{\bf Our Contributions.}
In this work we present a mechanism aimed at solving these issues. We propose to set short deadlines that are automatically extended if congestion occurs. We lay the theoretical foundations of congestion monitoring in blockchains and formalize the notion of challenge-response protocols in this context. 
We then propose two different protocols to detect congestion over multiple blocks: the `$L$ Consecutive Blocks' protocol defines uncongestion by the existence of $L$ consecutive uncongested blocks; its generalization, the `Sliding Window (K-out-of-N)' protocol, defines uncongestion by the existence of N consecutive uncongested blocks with K uncongested blocks among them.
We show that the Sliding Window protocol is more resilient to attacks than the $L$ Consecutive Blocks protocol when attacked by miners. Furthermore, we propose a new opcode for Ethereum that will provide the required functionality; we also provide an implementation (not requiring new opcodes) in Solidity, using opcodes introduced by Ethereum Improvement Proposal 1559 (EIP 1559)~\cite{buterin2019eip}.

\paragraph{Examples of Congestion Attacks and Related Failures.}
A recent well-known example of congestion-related failure took place on \emph{Crypto Black Thursday} (March 12th, 2020) when the price of Ethereum dropped by more than 50\% in less than 24 hours~\cite{cryptoblackthursday2}. This led to a panic-sale of coins and increased congestion. At the peak, during a 2-3 hour window, the Ethereum blockchain's fees climbed to \$1.65 on average, more than $10$ times their cost on previous days.

The drop in ETH price triggered many MakerDAO auctions to liquidate collateral (typically collateral on short positions must be sold if prices fluctuate too much). The tokens to be sold were purchased at almost no cost due to the inability of many bidders to send transactions and participate. 
This has, allegedly, been leveraged by one user to gain \$8.3 million worth of ETH~\cite{cryptoblackthursday1}.

Several studies~\cite{censorship_attacks,harris2020flood} deal with different types of attacks designed to prevent a party from responding on time to a challenge~\cite{censorship_attacks}. Harris and Zohar~\cite{harris2020flood} present an attack where the attacker forces many victims at once to flood the blockchain with claims for their funds. The ensuing congestion allows the attacker to steal the funds that cannot be claimed before the deadline.  
Our protocol will prevent these issues by extending the deadlines until the congestion passes.

\section{Related Work}\label{sec:relatedwork}
Congestion is a real-world problem faced by the most prominent cryptocurrencies.
In addition to the popular examples of Crypto Black Thursday or Cryptokitties, widely discussed online~\cite{cryptoblackthursday1,cryptoblackthursday2,consensys_2018}, Sokolov~\cite{sokolov2021ransomware} examined periods of congestion caused by ransomware.

One way to deal with congestion is to improve the scalability of the underlying consensus protocol~\cite{sompolinsky2018phantom,sompolinsky2015secure,sompolinsky2016spectre,croman2016scaling,eyal2016bitcoin} or to introduce higher-level layers that help to scale. Solutions ranging from sharding~\cite{wang2019sok}, off-chain payment channels~\cite{gudgeon2020sok} or layer-zero optimization~\cite{tanana2019avalanche}
(i.e., network-level optimization) have been considered. All these solutions improve the number of transactions per second that the network can process, but congestion may still occur even at higher rates.

Other methods that help to ensure that time-sensitive transactions are processed are rather ad-hoc. For example, the \emph{replace by fee}~\cite{replace-by-fee} and \emph{child pays for parent}~\cite{cpfp} mechanisms allow users to add or change the fees of their transactions. 
Bitcoin's fee mechanism---equivalent to a first-price auction---is sub-optimal and often results in users paying more than what is necessary. EIP 1559 was made to change this mechanism in Ethereum~\cite{roughgarden2020eip,buterin2019eip}.
EIP 1559 implements a \emph{base fee} that is burned. This base fee can be seen as an indication of the level of congestion in recent blocks, and we utilize this in our implementation. 

Another line of research that could potentially prevent transaction fees from spiking considers order-fairness consensus protocols~\cite{kelkar2020order,kursawe2020wendy}. The idea is to ensure that transactions are ordered in the blockchain in the same order they arrived in. This also helps to avoid problems such as front-running~\cite{daian2019flash}.

\section{Preliminaries and Definitions} \label{sec:FormalDef}
\subsection{Challenge-Response Protocols} \label{background:challenge-response}
A challenge-response protocol is an implementation of a pattern in which some party must respond to a challenge within a fixed time limit. This pattern consists of a challenge that takes effect at time \textbf{$T_c$} and a response deadline \textbf{$T_{rd}$} which is the latest time by which response to the challenge will be accepted. We call the time period between \textbf{$T_c$} and \textbf{$T_{rd}$} the challenge window. 
Responding to the challenge during the challenge window yields different results compared to responding \emph{after} the deadline.
The protocol we propose inspects the challenge window period and extends it (by extending \textbf{$T_{rd}$}) as long as the blockchain stays congested.

\subsection{Blockchain Congestion}\label{sec:definitions}
Our protocol has two components. First it relies on a mechanism to define what it means for a block to be congested. We then use this definition to define an uncongested \emph{period}. Intuitively, a period will be (un)congested if some threshold of blocks is (un)congested.
We start by defining block congestion before moving on to presenting different period congestion definitions and choosing one that meets our requirements.

\paragraph{Blocks and Transactions.} A block $\block = \{\tx_1,\cdots,\tx_n\}$ is as a set of transactions
(we ignore the order of transactions in the block as well as other data---such as nonce---as they are irrelevant to our problem).
Transactions pending to be included in a block are kept locally by each node in their \emph{mempool} until they are included in the chain.
Each transaction $\tx$ has a size \size{tx}, and a fee density \fee{tx}. The fee paid by the transaction is therefore $\size{tx}\cdot\fee{tx}$.
We define the total weight of transactions in a block $\block$ with a fee density above $\theta$ as  $\weight{\theta}{\block} \coloneqq \sum_{tx\in B:~\fee{tx}\geq \theta} \size{tx}$.

Blocks can contain transactions with total size bounded by \blocksizelimit, i.e., $\weight{\text{0}}{\block} \le \blocksizelimit$.  
For simplicity, 
we treat every block as \emph{full}, i.e., for any $block$ $\weight{\text{0}}{\block} = \blocksizelimit$ (if necessary,
we fill them artificially with transactions with a fee of 0).

The total amount of fees collected from a block by the miner is $\utility{\block} \coloneqq \sum_{tx\in B} \size{tx}\cdot\fee{tx}$. 
If the size of the mempool is bigger than the maximum block size $\weight{\text{0}}{\block}$, we assume that honest miners choose the transactions in a way to maximize the fees they get.

\paragraph{Period.}\label{def:period}
A period $Pe=(b_1,b_2,...,b_n)$ in the blockchain is a non-empty sequence of \textbf{consecutive} blocks. We denote the length (number of blocks) of the period by $|Pe|=n$, and write for $i\in\{1,...,n\}$: $Pe[i] = b_i\in Pe$.
For a period $P_2$, we say that period $P_1$ is included in $P_2$ and note $P_1\subseteq P_2$ if every block in $P_1$ is included in $P_2$.

In this work, we want to capture the notion of congestion: a phenomenon where there's a spike in the number of transactions waiting in the mempool. Since the mempool is not part of the blockchain, we instead rely on the data in the blocks in order to define congestion. We propose the following definition for \emph{block congestion}.

\begin{definition}[$(\theta,\gamma)$-congestion] \label{block_congestion_def}
	We say that a single block $\block$ is $(\theta,\gamma)$-conges-ted if
	$\weight{\theta}{\block} \geq \gamma \cdot \blocksizelimit$
	and denote $\IsBlockCongested{\theta,\gamma}{\block}=1$, where $\mathcal{C}_{\theta,\gamma}$ is the corresponding indicator function.
\end{definition}

Per this definition, all transactions above fee density $\theta$ are examined and required to make up at least a $\gamma$-fraction of the block in terms of size.
Intuitively, for $\gamma=1$ the definition captures that if a block is $(\theta,1)$-congested, a transaction needs to have a fee density that is at least $\theta$ in order to have a better chance of being included. 
In other words, we use the price of entering a transaction to the blockchain as a reliable signal for congestion.

For a block $\block$ and a fee density $\theta \geq 0$, we define the \emph{$\theta$-weight threshold} $\gamma_{\block}(\theta)$ as the maximum fraction of the block weight under which
the block is $(\theta,\gamma)$-congested. From definition~\ref{block_congestion_def} it is clear that $\gamma_{\block}(\theta) = \frac{\weight{\theta}{\block}}{\blocksizelimit}$.
Similarly, for a block $\block$ and a fraction $\gamma \geq 0$, we define the \emph{$\gamma$-fee density threshold} $\theta_{\block}(\gamma)$ as the maximum fee density under which the block is $(\theta,\gamma)$-congested ($\theta_{\block}(\gamma) \coloneqq \max\{\theta \mid \IsBlockCongested{\theta,\gamma}{\block}=1\}$).

\paragraph{Block Manipulation.} One of the key measures we are interested in is when is an adversary able to manipulate blocks' congestion signals. When a miner mines a block, they can choose to include transactions from their mempool or to add dummy transactions that move money between their accounts and pay a fee (to themselves), making the fees appear different than they ought to be.
However, miners cannot manipulate blocks at arbitrary heights, and doing so would incur a cost. 
The miner's chance of mining a new block depends on its relative computational power. Therefore, as is standard, we denote the computational power of an adversary by $\alpha$. Each block has a probability $\alpha$ to be mined by the adversary, and $1-\alpha$ to be mined by the other miners.
Furthermore, giving up mempool transactions means missing out their fees and hence induces a loss that we compute in the next two propositions.

\begin{proposition}\label{prop:one}

	An adversary manipulating a block $\block$ to make it $(\theta_1,\gamma_1)$-conges-ted when it is not, will lose a potential profit of at-least  $\blocksizelimit\cdot\int_{1-(\gamma_1-\gamma_{\block}(\theta_1))}^{1} \theta_{\block}(\gamma) \,d\gamma$.
\end{proposition}

\newcommand{\proofpropositionone}{
	\begin{figure}[ht]
		\centering
		\includegraphics[scale=0.48]{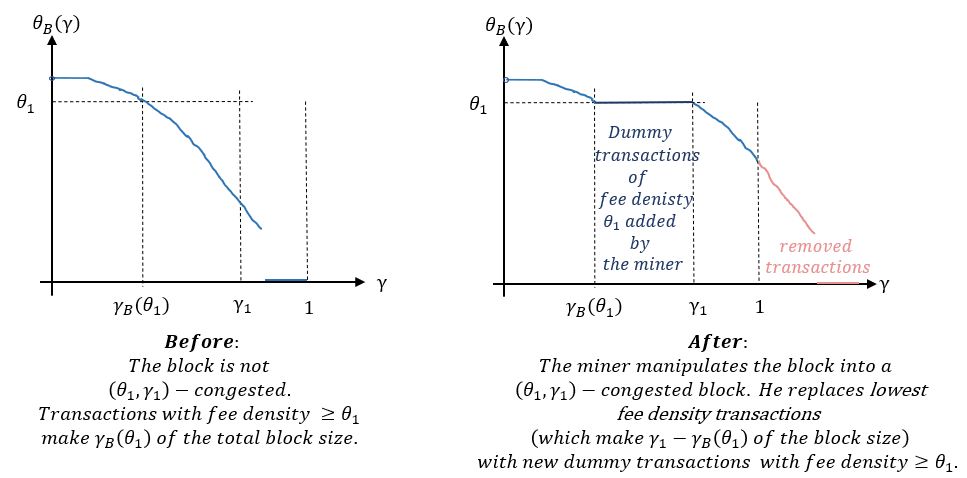}
		\caption{Cumulative function of the $\theta$-weight threshold of a block $\block$ before and after miner manipulation}
		\label{fig:tight1}
	\end{figure}

	\begin{proof}
		Given a block $\block$ which is \textbf{not} $(\theta_1,\gamma_1)$-congested ($\IsBlockCongested{\mathnormal{\theta_1},\mathnormal{\gamma_1}}{\block}=0$), it is possible to manipulate it into a block $\overline{\block}$ which is $(\theta_1,\gamma_1)$-congested ($\IsBlockCongested{\mathnormal{\theta_1},\mathnormal{\gamma_1}}{\overline{\block}}=1$) by replacing
		some of its transactions with dummy transactions that have fee density $\ge \theta_1$.
		In order to maximize its revenue, the adversary will remove the transactions with the lowest fee density.
		The minimum portion of transactions that the adversary needs to remove is $\gamma_1-\gamma_{\block}(\theta_1)$ (by the definition of
		$\gamma_{\block}$), and by doing so the miner misses the rewards associated with removing these legitimate transactions.
		
		Using the notations from above, we compute a lower bound on the miner's loss, which can be expressed by
		$\utility{\overline{\block}}\leq \utility{\block} - \blocksizelimit\cdot~\int_{1-(\gamma_1-\gamma_{\block}(\theta_1))}^{1} \theta_{\block}(\gamma) \,d\gamma$ (see Figure~\ref{fig:tight1}).
		Note this is only a lower bound since a miner can remove transactions only in their entirety and not parts of them.
\end{proof}}
\ifsubmission
\else
\proofpropositionone
\fi

\begin{proposition}\label{prop:two}
	An adversary manipulating a block $\block$ to reverse its signal from $(\theta_1,\gamma_1)$-congested to not congested will lose a potential profit of at-least  $\blocksizelimit\cdot\int_{\gamma_1}^{\gamma_{\block}(\theta_1)}(\theta_{\block}(\gamma)-\theta_1) \,d\gamma$.
\end{proposition}

\newcommand{\proofpropositiontwo}{
	\begin{figure}[ht]
		\centering
		\includegraphics[scale=0.48]{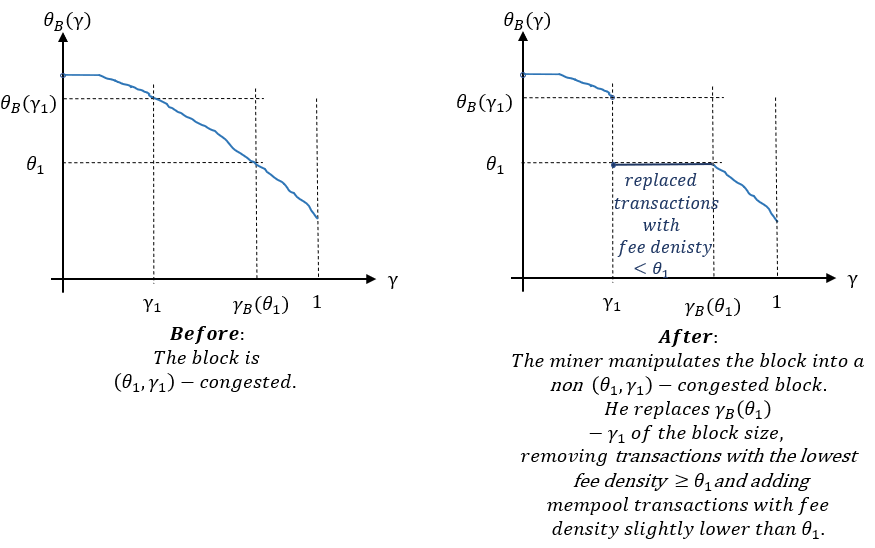}
		\caption{Cumulative function of the $\theta$-weight threshold of a block $\block$ before and after miner manipulation}
		\label{fig:tight2}
	\end{figure} 
	
	\begin{proof}
		The cost of reverting a block signal to uncongested depends on the state of the mempool. It is likely that transactions with fee density slightly lower than $\theta_1$ will be available in the miner's mempool allowing him to substitute the original transactions with fee density $\geq \theta_1$ for these and lose the fee difference. 
		This allows us to give a bottom bound on the loss which can be expressed by
		$\utility{\overline{\block}}\leq \utility{\block} - \blocksizelimit\cdot\int_{\gamma_1}^{\gamma_{\block}(\theta_1)}(\theta_{\block}(\gamma)-\theta_1) \,d\gamma$ (see Figure~\ref{fig:tight2}).
\end{proof}}

\ifsubmission{The proofs for both propositions can be found in Appendix~\ref{app:proofs1}.}
\else
\proofpropositiontwo
\fi

Before moving on to define period congestion, we note that there exist other ways in which block congestion could be defined. For example, in Section~\ref{sec:Implementation}, we take the EIP 1559 \emph{base fee} as a measure of congestion and use it to implement our suggested protocol. We include several other examples that are less efficient in Appendix~\ref{app:block_congestion_def_examples}.

\paragraph{Congestion Vector of a Period.} To determine whether a period $Pe$ is uncongested we will refer to the congestion vector $Pe^c \coloneqq(\IsBlockCongested{}{\mathnormal{Pe[i]}})_{i=1}^{n} \in \{0,1\}^n$ which consists of the congestion signal of its blocks. 
Intuitively, if most of the blocks in the period are congested then the period is congested and vice-versa. However, we must also account for the fact that an adversary may be able to change the congestion signal of some of the blocks, as already discussed.
We will consider different protocols to define period uncongestion, a situation in which the period is considered not congested. An uncongestion period protocol is a function that we denote by
$\IsPeriodUnCongested:\{0,1\}^*\rightarrow \{0,1\}$.
This function takes as input a binary series representing the congestion signal of the blocks in the examined time period.
It will return 0 if the period is congested and 1 otherwise.
This function can furthermore be extended to also provide auxiliary information such as a proof $\pi$ in the case where the period is uncongested. 
For the efficiency of the protocol, we will strive for a definition that can provide a compact and easy-to-verify proof.
Throughout the rest of the paper,
we use $B(n,p)$ to denote the binomial distribution with parameters $n$ and $p$.
\begin{definition}[Period Manipulation]
	For a period $Pe$ and an adversary with a fraction $\alpha$ of the total computational power, we associate a manipulated period $\hat{Pe}$ defined as follows.
	For $i\in\{1,\dots,|Pe|\}$ the adversary can replace $Pe[i]$ with probability $\alpha$, with a block that has a congestion signal of their choice.
	We denote by $\overline{m}=m_{|Pe|}(\alpha) \sim B(|Pe|,\alpha)$ the vector that indicates which of the blocks in the given period the adversary controls, meaning the adversary can replace the $Pe[i]$ block's congestion signal iff $\overline{m}[i]=1$.
	We then define the adversary's manipulation set $S_{\overline{m},Pe} \coloneqq \{\hat{Pe^c} \in \{0,1\}^{|Pe|}\mid \forall 1\leq i \leq |Pe|: \overline{m}[i]=0 \Rightarrow \hat{Pe^c}[i] = Pe^c[i]\}$. Intuitively, $S_{\overline{m},Pe}$ corresponds to the set of possible congestion vectors that the adversary could create by changing the signal of the blocks that it controls.
\end{definition}

In a real world setting, even if there is a long period of uncongestion, it could be the case that one or more of the blocks are fuller than the others due to some randomness in the transactions' arrival time (e.g., there was a temporary high transaction volume).
To account for this randomness,
we make a simplifying assumption that blocks are congested independently with probability $p$ and say that the blockchain is \emph{$p-$congested}. We note that, in reality, congestion is often changing and is usually correlated when considering several consecutive blocks. We leave more complex models of congestion for future work. In our case, the congestion vector of a period $Pe$ chosen at random has a binomial distribution: $Pe^c \sim B(n,p)$.
When studying attacks where the adversary tries to convert a congested period to an uncongested one, we will assume that $p$ is close to one (i.e., most of the blocks are congested),
whereas when studying the opposite case, we will consider $p$ to be close to zero.

Our protocol consists in extending the deadline of challenge-response in the event of a congestion period.
However, to avoid an edge case where the deadline is extended indefinitely, we define $\hat{M}$, an upper bound on the total length of the extended period.

\begin{definition}[$\hat{M}$-maximum Extension] \label{def:M}
	Given a challenge-response protocol in a $p-$congested blockchain
	where the challenge starts at block height $h$, we say that $\hat{M}$ is the maximum extension of the challenge if the deadline cannot be extended further than height $h + \hat{M}$.
\end{definition}

\subsection{Desirable Properties of Protocols}\label{desirableProperties}

In this section, we define some properties that we aim for our protocol to achieve.

In the rest of the paper we use the notation $D \xleftarrow{} s$ to denote that $s$ was selected randomly from the distribution $D$.
We start by describing the two types of attack that we will consider---a congestion attack and an uncongestion attack---before defining the
\emph{robustness} of the protocol, which captures the security of the protocol against either attack. 

\begin{definition}[Congestion/Uncongestion Attack on $Pe$]\label{def:attacks}
	Given a period $Pe$, chosen at random in a $p$-congested blockchain,
	we say that the adversary wins a congestion, resp. uncongestion, attack on $Pe$
	if it can manipulate $Pe$ into an congested, resp. uncongested, period.
\end{definition}

\begin{definition}[$(\alpha, p, q, n)$-congestion Robustness]
	\label{def:con_rob}
	We say an uncongestion period protocol $\IsPeriodUnCongested:\{0,1\}^*\rightarrow \{0,1\}$ is \textbf{$(\alpha, p, q,n)$-congestion robust} if, given an adversary with a relative computational power $\alpha$, 
	his probability of winning a congestion attack, i.e., of successfully manipulating a period $Pe$ of $n$ blocks into a congested period $\hat{Pe}$, is less than $q$.
	$$ B(n,p)\xleftarrow{} Pe~:~
	P_r(\exists \hat{Pe} \in S_{\overline{m}, Pe}~s.t.~\IsPeriodUnCongested(\hat{Pe})=0) \leq q.$$
\end{definition}

\begin{definition}[$(\alpha, p, q, n)$-uncongestion Robustness]
	\label{def:uncon_rob}
	We say an uncongestion period protocol $\IsPeriodUnCongested:\{0,1\}^*\rightarrow \{0,1\}$ is \textbf{$(\alpha, p, q,n)$-uncongestion robust} if, given an adversary with a relative computational power $\alpha$, 
	his probability of winning an uncongestion attack, i.e., of successfully manipulating a period $Pe$ of $n$ blocks into an uncongested period $\hat{Pe}$, is less than $q$.
	$$ B(n,p)\xleftarrow{} Pe~:~
	P_r(\exists \hat{Pe} \in S_{\overline{m}, Pe}~s.t.~\IsPeriodUnCongested(\hat{Pe})=1) \leq q.$$
	
\end{definition}

\begin{definition}[Monotonicity]
	A congestion protocol is \textbf{monotone} if
	for every two periods $Pe_1$ and $Pe_2$, if $Pe_1\subseteq Pe_2$ and $Pe_1$ is considered uncongested, then so is $Pe_2$, i.e., 
	$\forall Pe_1\subseteq Pe_2: \IsPeriodUnCongested(Pe_1^c)=1 \rightarrow \IsPeriodUnCongested(Pe_2^c)=1$.
\end{definition}
A monotone protocol is easier to verify
as the prover only needs to select a portion of blocks from the time period $Pe$ in order to prove uncongestion.
Furthermore, a monotone protocol requires only sporadic access to the blockchain. A prover can go offline and prove uncongestion when they come back online by choosing any uncongested period from the time they were offline.
In the case of a non-monotonic protocol, if the prover is offline during an uncongested period, they cannot prove the uncongestion of the longer period, after they came back online, they missed the uncongested period.

\paragraph{Efficiency Properties.}
We define two properties that capture the efficiency of the protocol.
\begin{itemize}
	\item \textbf{Concise proof size}  The evidence needed to prove uncongestion of a period should be as concise as possible.
	\item \textbf{Concise refresh information} The extra information needed to be kept when checking the congestion signal of a period that
	has already been extended due to congestion should be as concise as possible. Ideally, when we extend a period from $Pe_1$ to $Pe_2$ in order to check $Pe_2$ for congestion, we should not have to recheck every block in $Pe_1$ but, rather, aggregate this information.
\end{itemize}

In the next section we will discuss different period congestion protocols with the goal of finding one that will be proof efficient and robust against an attacker with reasonable hash rate with high probability.

\section{Uncongested Period Protocols}\ \label{sec:UncongestedPeriodProtocols}
In this section, we examine different protocols that fit the definition of congestion of a period $Pe$.
We start by presenting ``naive'' protocols and discuss why they are not good enough, i.e., why they lack the desirable properties defined in Section~\ref{desirableProperties}. 

\subsection{Strawman Protocols}

\begin{definition}[Cumulative M]
	Period $Pe$ is uncongested if there exists M blocks which are uncongested:
	$\UCP{CM}{Pe^c}=1 \leftrightarrow \left(\sum_{b\in Pe} (1-\IsBlockCongested{}{\mathnormal{b}})\geq M\right)$.
\end{definition}
This protocol is monotonic but is not sufficiently robust to adversarial attacks:
if we wait long enough, the probability of the adversary controlling M blocks becomes overwhelming (even if $\alpha$ is small). We solve this in the next strawman by considering the percentage of blocks instead of a fixed number.

\begin{definition}[Percentage]
	A period $Pe$ is uncongested if $x\%$ of its blocks are not congested:
	$\UCP{PC}{Pe^c}=1 \leftrightarrow \left(\sum_{b\in Pe} (1-\IsBlockCongested{}{\mathnormal{b}})\geq \frac{x}{100}\cdot|Pe|\right) $.
\end{definition}
This protocol is much more robust but has the drawback of not being monotonic. For example, if all blocks are uncongested during the first part of the period and congestion begins in the second part, then the beginning of the period is uncongested while the whole period may not be. 

We now suggest the following monotonic rule:
\begin{definition}[$L$ Consecutive Blocks]
	A period $Pe$ is uncongested if there exists at least L
	consecutive uncongested blocks included in it:\\
	$\UCP{L}{Pe^c}=1 \leftrightarrow \left(\exists~ 1 \leq i \leq |Pe|-L+1~s.t.~\forall~0\leq j \leq L-1:Pe^c[i+j]=0\right) $.
\end{definition}
\newcommand{\proofconscuivemonotone}{
	\begin{newProposition}{3}
		The $L$ Consecutive Blocks protocol is monotonic.
	\end{newProposition}
	
	\begin{proof}
		Given a period $Pe_1$,
		uncongested according to the $L$ Consecutive Blocks protocol, which is included in period $Pe_2$:
		\begin{align*}
			&\UCP{L}{Pe_1^c}=1\Rightarrow \\
			& \exists~ 1 \leq i_1 \leq |Pe_1|-L+1~s.t.~\forall~0\leq j \leq L-1:Pe_1^c[i_1+j]=0 \\
			& Pe_1\subseteq Pe_2 \Rightarrow \\  
			& \exists~ 1 \leq k \leq |Pe_2|-|Pe_1|+1~s.t.~\forall~1\leq d \leq |Pe_1|:Pe_1^c[d]=Pe_2^c[k+d-1]\\
			&\Rightarrow for~ i_2=k+i_1-1,~\forall~0\leq j \leq L-1:Pe_2^c[i_2+j]=0 \\
			& \Rightarrow \UCP{L}{Pe_2^c}=1
		\end{align*}
\end{proof}}

\ifsubmission{
	We show that this protocol is monotonic and inspect its efficiency in Appendix~\ref{app:proofs2}.
	We now evaluate its robustness. }
\else{We first show that this protocol is monotonic, before evaluate its efficiency.
	\begin{proposition}
		The $L$ Consecutive Blocks protocol is monotonic.
	\end{proposition}
	\proofconscuivemonotone}
\fi

\newcommand{\LconsecEffInfo}{
	We now evaluate its efficiency.
	\paragraph{Proof size.}
	In order to provide evidence for the uncongestion of period $Pe$ of size $n$, it is enough to point to the location of the first block in a series of $L$ consecutive uncongested blocks, given formally by $\pi=\min\{i\in\{1,...,n-L+1\} \mid \forall~0\leq j \leq L-1:Pe^c[i+j]=0\}$.

	\paragraph{Refresh Information.}
	Given a congested period $Pe$, and $\hat{Pe}$ that extends it,
	in order to determine the congestion level of the extended period $\UCP{L}{\hat{Pe^c}}$, it is enough to check only $\hat{Pe} \setminus\left( Pe[:-(L-1)]\right)$, the period beginning $L$-1 blocks before $Pe$ ends.
}
\ifsubmission{
}\else{\LconsecEffInfo
}
\fi

\subsubsection*{Evaluation of the Robustness of the $L$ Consecutive Blocks protocol.}

We examine situations where the adversary attempts to manipulate the congestion signal for a given period. We separate this into two attacks: uncongestion and congestion attacks (as in Definition~\ref{def:attacks}).
We strive to achieve a high defense rate against both attacks, meaning finding a value $L$ that will give a low probability for an adversary to succeed in each of the attacks separately.

\paragraph{Evaluation of the Uncongestion Attack.}
In order to compute the probability of an attacker to successfully manipulate $Pe$ into an uncongested period, we define the following matrix $T_{(L+1)\times(L+1)}$:
\begin{equation}
	\forall~0 \leq i,j\leq L:~~ T_{i,j} =
	\begin{cases*}
		(1-\alpha)\cdot p & if $j = 0 \wedge i\neq L$ \\
		\alpha + (1-\alpha)\cdot (1-p) & if $j = i+1$ \\
		1 & if $i = j = L$ \\
		0        & otherwise
	\end{cases*}
\end{equation}

and denote by $e_i$ the $i^{th}$ unit vector of dimension $L$+1 (i.e., $e_i$ has a 1 in the $i^{th}$ coordinate and 0's elsewhere).
\begin{theorem}
	The probability of an attacker with a relative computational power $\alpha$ to successfully manipulate $Pe$ into an uncongested period in a p-congested network
	equals    $e_1 \cdot T^n \cdot e_{L+1}^t$.
\end{theorem}
\begin{proof}
	We note that, at each block, the attacker has a probability $\alpha$ to mine the next block, which allows them to decide its congestion level. In this context, this means setting the block to be uncongested. In addition, the congestion signal of a block not mined by the attacker depends on the prevailing congestion state which is expressed by $p$. The probability of an honest block being congested, resp. uncongested, is hence equal to $(1-\alpha)\cdot p$, resp.  $\alpha+(1-\alpha)\cdot(1-p)$. We define the following Markov chain that describes a random walk on $Pe$'s blocks and whose states represent the number of consecutive blocks that are uncongested at a point in time.
	
	\begin{figure}[ht]
		\centering
		\includegraphics[scale=0.6, trim={200cm 0.8cm 200cm .5cm}]{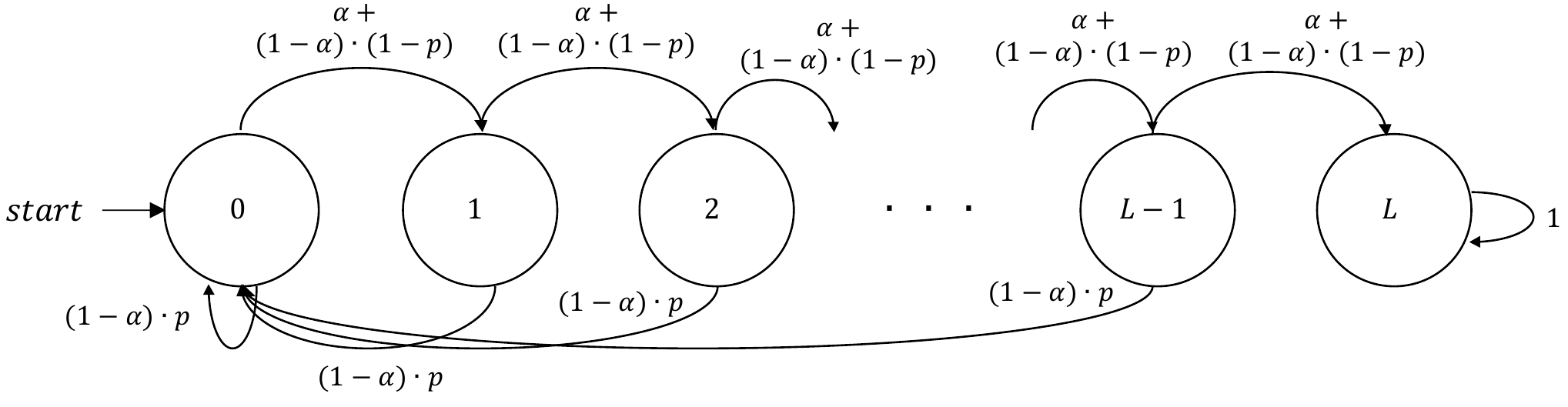}
	\end{figure} \label{fig:markov1}
	
	The initial state is $0$ since it corresponds to the 0 consecutive uncongested blocks at the beginning of the walk.
	With each step, we move from state $i$ to state $i+1$, for $i<L$, if the block is uncongested, and return to state $0$ if it is not.
	If we reach state $L$, we stay there since it means the adversary has reached the goal of $L$ consecutive uncongested blocks in $Pe$ and can manipulate it to an uncongested period.

	T is the corresponding transition matrix; hence the probability of reaching state $L$ in $|Pe|=n$ steps is expressed by $e_1 \cdot T^n \cdot e_{L+1}^t$.
\end{proof}

\paragraph{Evaluation of the Congestion Attack.}
For the attack in the opposite direction we 
define $\hat{T}_{(L+1)\times(L+1)}$ as follows:
\begin{equation}
	\forall~0 \leq i,j\leq L:~~ \hat{T}_{i,j} =
	\begin{cases*}
		\alpha + (1-\alpha) \cdot p & if $j = 0 \wedge i\neq L$ \\
		(1-\alpha) \cdot (1-p) & if $j = i+1$ \\
		1 & if $i = j = L$ \\
		0        & otherwise
	\end{cases*}
\end{equation}

\begin{theorem}
	The probability of an attacker with a relative computational power $\alpha$ to successfully manipulate $Pe$ into a congested period, in a p-congested network
	equals    $1-e_1 \cdot \hat{T^n} \cdot e_{L+1}^t$.
\end{theorem}
\newcommand{\prooflconsecutivecongestion}{
	\begin{proof}
		This time, if the attacker succeeds in mining a block, they will make it congested. Therefore the probability for a block to be uncongested is $(1-\alpha)\cdot(1-p)$. 
		As before, we define the following Markov chain whose states represent the number of consecutive blocks that are uncongested at a point in time in $Pe$:
		
		\begin{figure}[ht]
			\centering
			\includegraphics[scale=0.6, trim={200cm 0.8cm 200cm .5cm}]{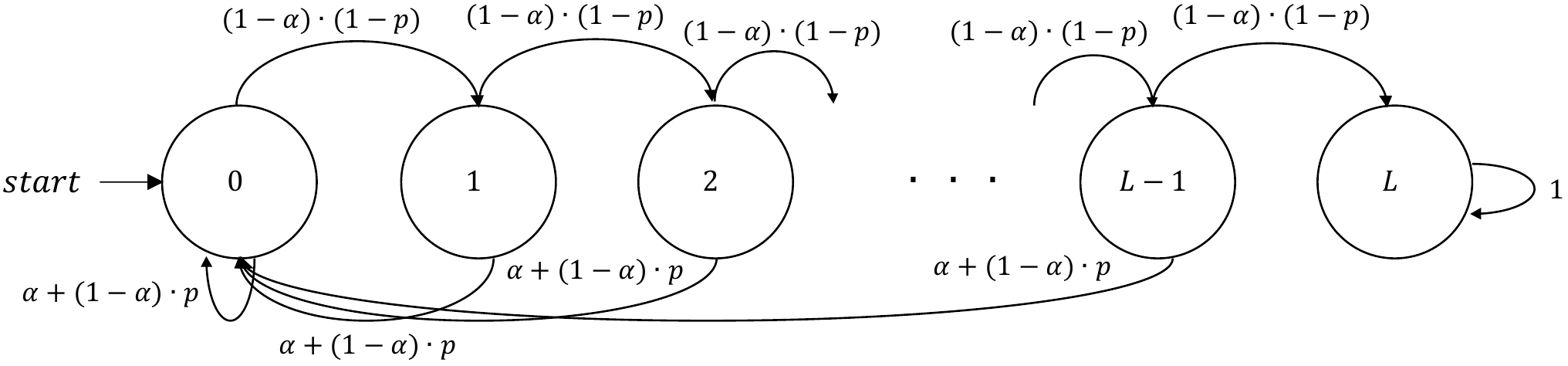}
		\end{figure} \label{fig:markov2}
		
		$\hat{T}$ is the corresponding transition matrix.
		Therefore, the probability for the adversary to succeed in the congestion attack is equivalent to the probability that the rest of the miners will not reach the $L$ state in $n$ steps, which is expressed by: $1-e_1 \cdot \hat{T^n} \cdot e_{L+1}^t$.
\end{proof}}
\prooflconsecutivecongestion
% \ifsubmission{The proof in this case works very similarly as in the previous theorem. We leave the details in the appendix (Appendix~\ref{app:proofs2}).}
% \else{\prooflconsecutivecongestion}\fi

Now that we have the attacks' success rates, we examine the robustness of the protocol against both attacks for different values of $L$.

Although attacks are potentially expensive for the adversary (who needs to change the contents of its block and, hence, loses transaction fees), we still desire a low success probability for the attack even for strong attackers. 
We assume in the following evaluations that the attacker controls $33\%$ of the computational power.

Given that congestion may cause period extension, we need a value for $L$ that gives protection also against attacks over longer periods. We examine the behavior of the protocol for periods as long as $\hat{M}$ using different values for $L$.

The value $p$ should represent realistic network conditions. For our analysis we pick $p=0.85$ when studying the uncongestion attack, to simulate more congested settings, or $p=0.15$ when studying the congestion attack, to simulate relatively uncongested settings. 
Other values can be plugged in, if needed, for other conditions. We start by examining the robustness of the protocol for a period of 1 day.

Figures~\ref{fig:consec-eth}-\ref{fig:consec-bitcoin} present 
the probability of success in both attacks for two different period lengths: 6450 blocks in Figure~\ref{fig:consec-eth} and 144 blocks in Figure~\ref{fig:consec-bitcoin}. These periods correspond, roughly, to a single day in Ethereum and a single day in Bitcoin. The red curves correspond to the congestion attack and the blue curves to the uncongestion attack. 
We compute these probabilities for different values of $L$.

\begin{figure}[ht]
	\begin{subfigure}[t]{.5\linewidth}
		\includegraphics[scale=0.38]{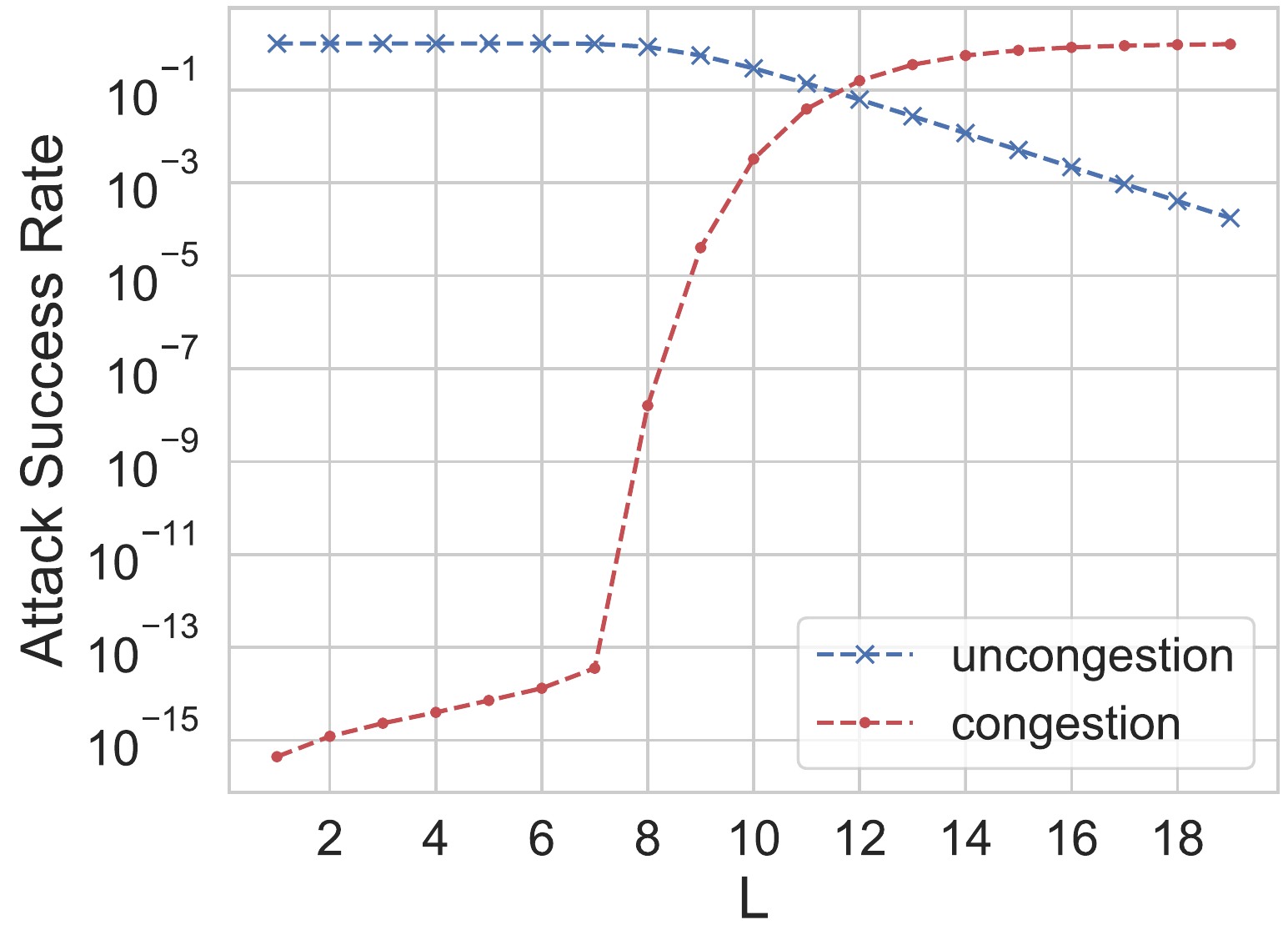}
		\caption{$|Pe|=6450\sim 1$ day in Ethereum}
		\label{fig:consec-eth}
	\end{subfigure}
	\label{fig:consec}
	\begin{subfigure}[t]{.5\linewidth}
		\includegraphics[scale=0.38] {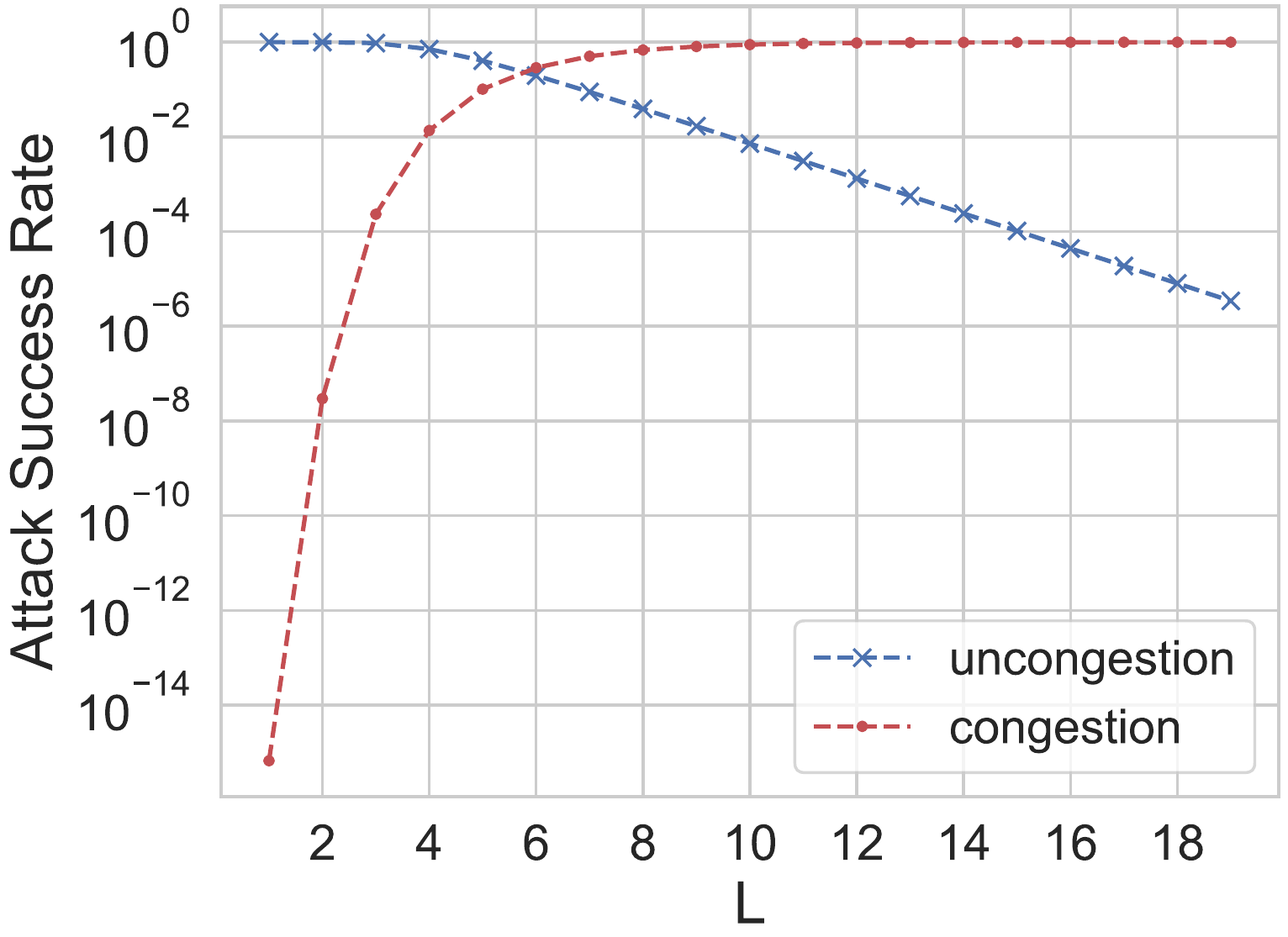}
		\caption{$|Pe|=144\sim 1$ day in Bitcoin} 
		\label{fig:consec-bitcoin}
	\end{subfigure}
	\caption{Attack success rate as a function of $L$, for $\alpha=0.33$}
\end{figure}

The results in both figures show there is no value $L$ that gives a probability of success less than $1\%$ for both attacks. Formally, it shows that the $L$ Consecutive Blocks protocol cannot be simultaneously $(0.33, 0.15, 0.01, 1~day)$-congestion~robust and
$(0.33, 0.85, 0.01, 1~day)$-uncongestion~robust.
Therefore,
we find the $L$ Consecutive Blocks protocol not sufficiently secure. Intuitively, this is because more robust estimates of congestion are typically obtained over longer observation windows. The $L$ Consecutive Blocks protocol obtains longer observations if $L$ is increased, but then the requirement for consecutive blocks to be uncongested is too strict and is not robust. As a result of this insight, we propose a new protocol that generalizes the $L$ Consecutive Blocks protocol and allows for longer observation windows with a relaxed condition for uncongestion.

\subsection{\textbf{Sliding Window (K-out-of-N) Protocol}}\label{sliding}

\begin{definition}[K-out-of-N Sliding Window] A period $Pe$ is uncongested if there exists a period $\hat{Pe}$ of length $N$ included in it in which at least $K$ blocks are uncongested.
	
	$$\UCP{SW}{Pe^c}=1 \leftrightarrow\left( ~\exists~\hat{Pe}\subseteq Pe: |\hat{Pe}|=N \wedge \left(\sum_{b\in \hat{Pe}} (1-\IsBlockCongested{}{\mathnormal{b}})\geq K\right)\right)$$
\end{definition}

We note that the $L$ Consecutive Blocks protocol is a special case in which $L=N=K$.
\newcommand{\slidingwindowmonotonicity}{
	\begin{newProposition}{4}
		The Sliding Window protocol is monotonic.
	\end{newProposition}
	\begin{proof}
		Given an uncongested period $Pe_1$, according to the Sliding Window protocol, which is included in period $Pe_2$:
		\begin{align*}
			& \UCP{SW}{Pe_1^c}=1\Rightarrow \left(\exists~\hat{Pe}\subseteq Pe_1: |\hat{Pe}|=N \wedge \left(\sum_{b\in \hat{Pe}} \IsBlockCongested{}{\mathnormal{b}}\geq K\right)\right)\\
			& Pe_1\subseteq Pe_2 \Rightarrow  \left(\hat{Pe}\subseteq Pe_2\right) \wedge \left(\sum_{b\in \hat{Pe}} \IsBlockCongested{}{\mathnormal{b}}\geq K\right)\\
			&\Rightarrow \UCP{SW}{Pe_2^c}=1
		\end{align*}
\end{proof}}
\slidingwindowmonotonicity
% \ifsubmission{
% 	In Appendix~\ref{app:proofs2}, we show that the Sliding Window protocol is monotonic.}
% \else{\slidingwindowmonotonicity}\fi

We now evaluate its efficiency.
\paragraph{Proof Size.}
In order to provide evidence for the uncongestion of a period $Pe$ of size $n$, it is enough to point to a window in which uncongestion occurs. Formally, to present $\pi=i \in \{1,...,n-K+1\}$ s.t. $\sum_{l=i}^{i+N}(1-\IsBlockCongested{}{\mathnormal{Pe[l]}})\geq K$.

\paragraph{Refresh Information.}
Given a congested period $Pe$, and $\hat{Pe}$ that extends it, 
in order to determine the congestion level of the extended period $\UCP{SW}{\hat{Pe^c}}$, it is enough to check only windows that overlap blocks in $\hat{Pe} \setminus Pe$.

\subsubsection*{\textbf{Evaluation of the Sliding Window Protocol's Robustness.}}
We consider the two attacks in Definition~\ref{def:attacks}.
We first note that the two attacks may differ in their consequences. While the congestion attack can cause a delay in the response deadline (i.e., a deadline will be extended even if it is not really needed), the uncongestion attack might lead participants to miss the chance to respond on time,
as the deadline will not be extended even if the network is congested. The damage in each case depends on the particular use case. 
For example, in the case of payment channels, not responding in time is more severe and may lead to financial losses. 
We strive to achieve a high level of security against both types of attack, i.e., to find values for parameters ($N, K$) that will yield a low probability of success for both.

We begin by presenting upper bounds on the probabilities of success in each of the attacks.
% \ifsubmission{The proofs are left in Appendix~\ref{app:proofs2}.}
% \fi
\begin{theorem}
	\label{thm:bound_un}
	The probability of an attacker with a relative computational power $\alpha$ to successfully manipulate $Pe$ into an uncongested period, in a p-congested network, is bounded above by $(n-N+1)\cdot \sum_{j=K}^{N}\binom{N}{j}\cdot q^{j}\cdot(1-q)^{N-j}$, for $q=\alpha+(1-p)\cdot(1-\alpha)$.
\end{theorem}

\newcommand{\proofslidingwindowuncongestion}{
	\begin{proof}
		The probability for a block to be uncongested during this attack is $q=\alpha+(1-p)\cdot(1-\alpha)$.
		In a period of size $n$, there are $n-N+1$ different sliding windows. We denote by $A_i$ the event in which there are $K$ out of $N$ uncongested blocks in the $i^{th}$ sliding window.
		Therefore, the probability of a single sliding window being uncongested is $P(A_i)=\sum_{j=K}^{N}\binom{N}{j}\cdot q^{j}\cdot(1-q)^{N-j}$. 
		To succeed in the uncongestion attack, at least one of the sliding windows has to be uncongested, which is expressed by $P(\cup_{i=1}^{n-N+1}A_i)$.
		We use the union bound to bound this probability and get: \\ $P(\cup_{i=1}^{n-N+1}A_i) \leq \sum_{i=1}^{n-N+1}P(A_i) = (n-N+1)\cdot \sum_{j=K}^{N}\binom{N}{j}\cdot q^{j}\cdot(1-q)^{N-j}$
\end{proof}}
\proofslidingwindowuncongestion
% \ifsubmission
% \else{\proofslidingwindowuncongestion}
% \fi
\begin{theorem}
	\label{thm:bound_co}
	The probability of an attacker with a relative computational power $\alpha$ to successfully manipulate $Pe$ into a congested period, in a p-congested network is bounded above by $(\sum_{j=0}^{K-1}\binom{N}{j}\cdot q^{j}\cdot(1-q)^{N-j})^{\lfloor \frac{n}{N} \rfloor}$, for $q=(1-p)\cdot(1-\alpha)$.
\end{theorem}
\newcommand{\slidingwindowcongestion}{
	\begin{proof}
		The probability for a block to be uncongested is $q=(1-p)\cdot(1-\alpha)$.
		We denote by $B_i$ the event in which there are less than $K$ uncongested blocks in the $i^{th}$ sliding window.
		The probability of a single sliding window being congested is $P(B_i)=\sum_{j=0}^{K-1}\binom{N}{j}\cdot q^{j}\cdot(1-q)^{N-j}$. 
		To succeed in the congestion attack, all sliding windows in the period must be congested, which is expressed by $P(\cap_{i=1}^{n-N+1}B_i)$.
		We bound this probability by $P(\cap_{i=1}^{\lfloor \frac{n}{N} \rfloor}B_{N \cdot (i-1) + 1})$, i.e., we consider a subset of events $B_i$ that are independent from each other (removing overlapping windows). We compute the intersection of the pairwise independent events and get: $P(\cap_{i=1}^{n-N+1}B_i) \leq P(\cap_{i=1}^{\lfloor \frac{n}{N} \rfloor}B_{N \cdot (i-1) + 1}) = \prod_{i=1}^{\lfloor \frac{n}{N} \rfloor}P(B_{N \cdot (i-1) + 1}) = (\sum_{j=0}^{K-1}\binom{N}{j}\cdot q^{j}\cdot(1-q)^{N-j})^{\lfloor \frac{n}{N} \rfloor}$.
\end{proof}}
% \ifsubmission
% \else{\slidingwindowcongestion}
% \fi
\slidingwindowcongestion

We would like to compute the robustness of the protocol for 1 day to 1 hour sliding windows.
We examine the situation where a period $Pe$ of size $n$ is chosen at random and the blockchain is $p-congested$ for values of $p=0.85$ (relatively congested) and $p=0.15$ (relatively uncongested) against an attacker with computational power $\alpha\leq0.33$.
In the evaluation, we allow periods to be extended up to two weeks, a reasonable time for congestion to pass. We set the $\hat{M}$-maximum extension (see Definition~\ref{def:M}) accordingly (90300 blocks in Ethereum and 2016 blocks in Bitcoin).

We first evaluate the attack over Ethereum, computing the above bounds for different sliding window sizes.
We begin with a sliding window of 1 day ($N=6450$), setting $K=\frac{N}{2} = 3225$. 
Figure~\ref{fig:sliding2w_upper_eth} presents the two upper bounds for the different possible period lengths $N\leq n \leq \hat{M}$. For the protocol to be considered secure, we need low values in both curves for the different period lengths (since periods might be extended).
As can be seen, the probabilities in the graph are extremely low, showing the protocol to be very secure. We emphasize that the blue curve is not horizontal, as shown in the graph; all of its values are smaller than $10^{-323}$. Note that these are only upper bounds; the actual probabilities are even lower.

\begin{figure}[ht]
	\centering
	\includegraphics[scale=0.43]{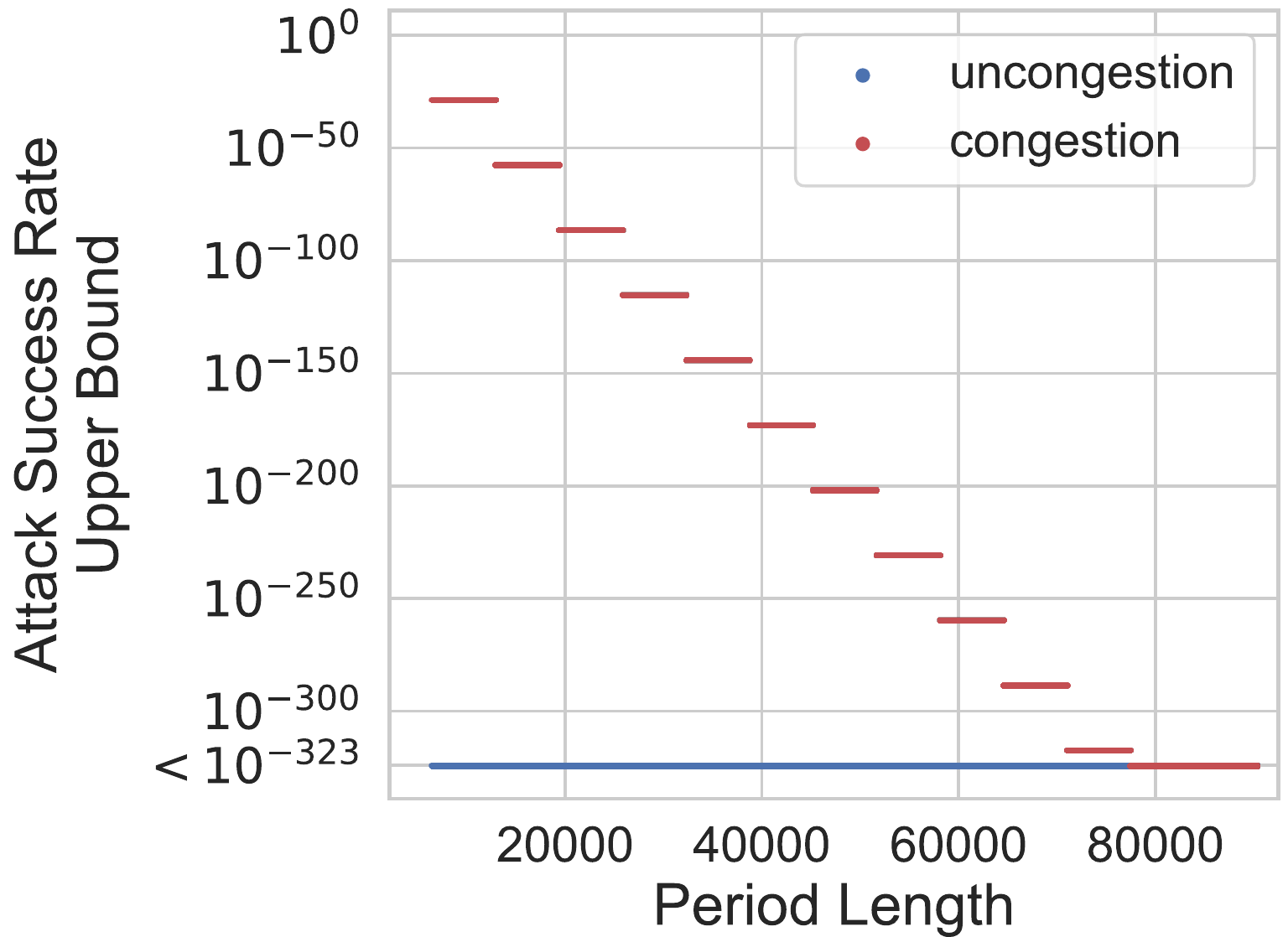}
	\caption{Upper bounds on the attacks' success rates as a function of the period length, for $\hat{M}=90300, N = 6450, K=3225, \alpha=0.33$}
	\label{fig:sliding2w_upper_eth}
\end{figure} 

We evaluate the attack for smaller sliding windows. The following table 
summarizes our results: 

\begin{center}
	\setlength{\tabcolsep}{10pt}
	\resizebox{9,2cm}{!}{%
		\begin{tabular}{ c  c  c  c }
			\hline
			\textbf{N} & \textbf{K} & \textbf{Uncongestion} & \textbf{Congestion}\\ 
			\hline
			6450 (1 day)   & 3225        & $<10^{-323}$   & $1.44\times 10^{-29}$  
			\\ 
			3225 (12 hours)   & 1612    & $1.26\times 10^{-10}$   & $8.06\times 10^{-16}$ 
			\\
			1612 (6 hours)   & 815       & $7.14\times 10^{-5}$   & $1.08\times 10^{-7}$ 
			\\	
			806 (3 hours)   & 421       & $8.87\times 10^{-3}$   & $3.16\times 10^{-3}$ 
			\\ \hline
		\end{tabular}
	}
\end{center}

The wider the sliding window is, the greater the protection.  
For smaller sliding windows, such as 1 hour ($N=269$), we can achieve a $99\%$ defense rate against each attack if we lower the attackers' computation power to $\alpha \leq 0.27$ (instead of 0.33).
We provide examples of $N, K$ values and the level of protection they provide (an upper bound), but these are configurable and subject to the user's discretion. One can choose to increase the level of protection from one attack at the expense of the other, or to set a larger initial period length ($>N$) to increase the protection.

Next, we want to know what happens with smaller periods such as in Bitcoin, which has longer block intervals. To do so, we set $\hat{M}=2016$ and begin with a sliding window of 1 day ($N=144$).

We use a simulation to draw 100,000 samples $Pe^c \sim B(\hat M,p)$ of congestion vectors and to compute the success rates of both attacks among the samples (Figures~\ref{fig:sliding2w}a,c,d). We use error plots to plot the standard error of the data; however, the errors are very small and therefore are almost invisible in the graphs. 

\begin{figure}[ht]
	\begin{subfigure}[t]{.5\linewidth}
		\centering
		\includegraphics[scale=0.39] {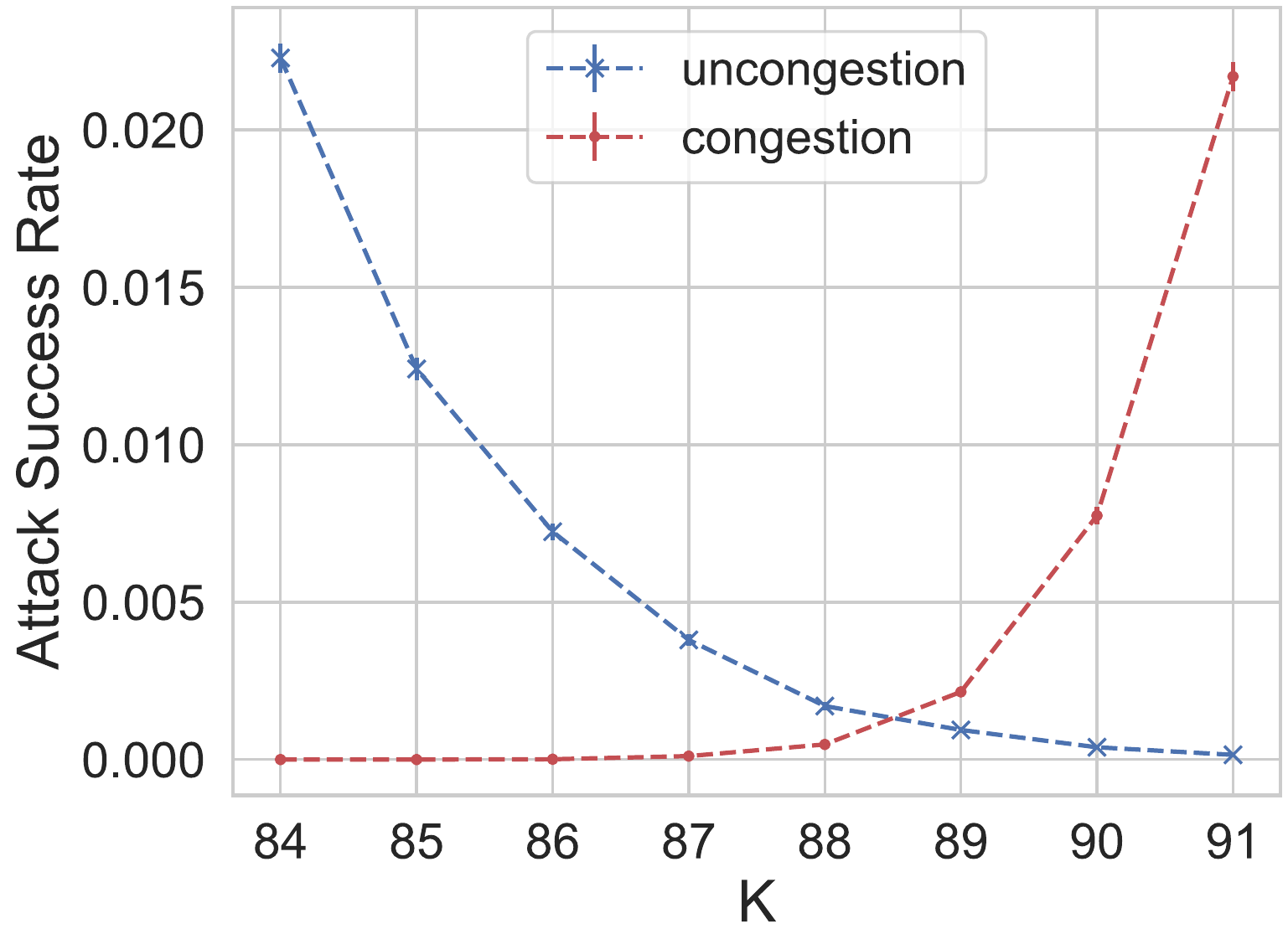}
		\caption{Attack success rate as a function of K, for $n = 2016, N = 144,\alpha=0.33$} 
		\label{fig:sliding2w-a}
	\end{subfigure}
	\hspace{1em}%
	\begin{subfigure}[t]{.5\linewidth}
		\centering
		\includegraphics[scale=0.38]{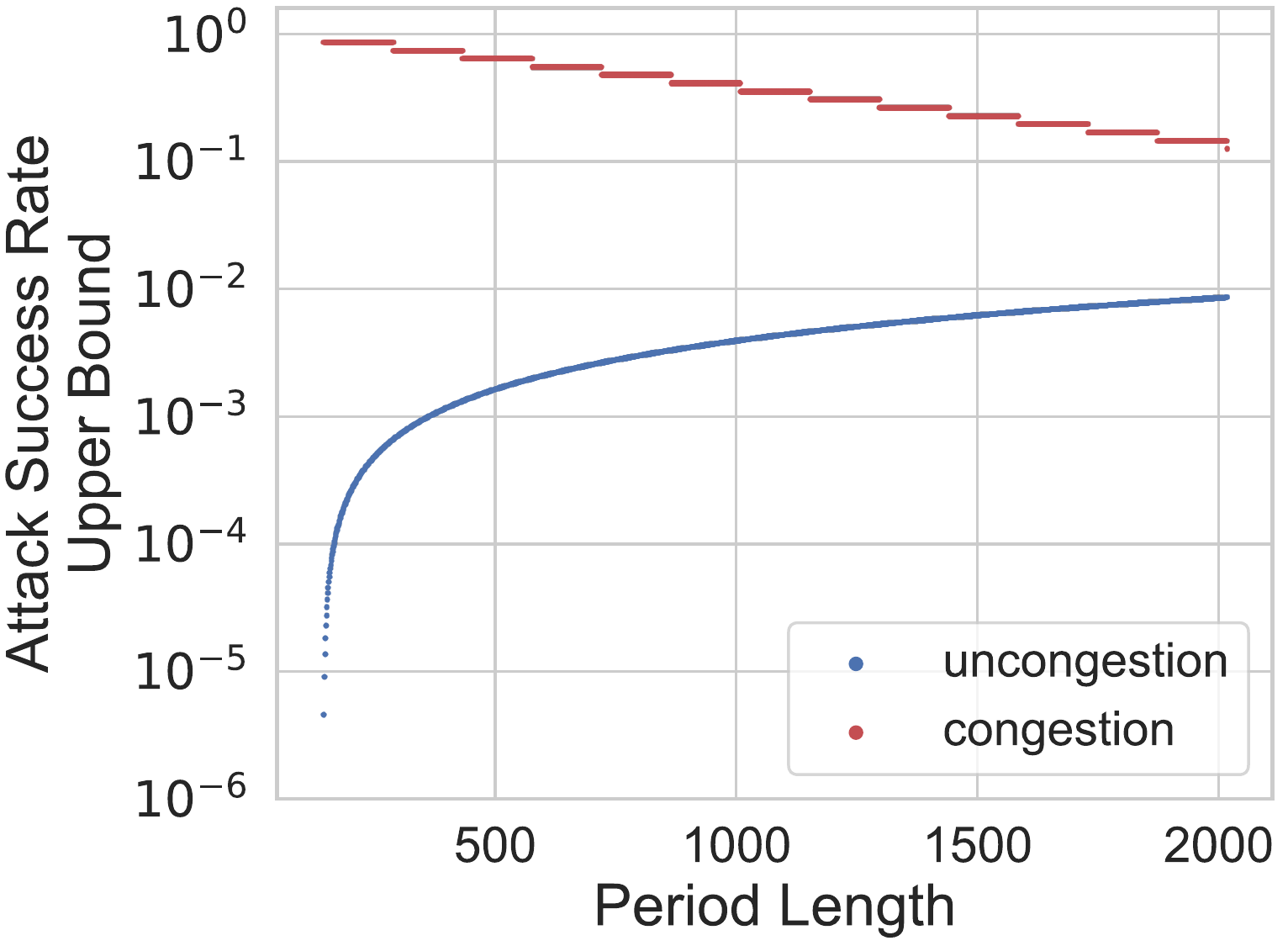}
		\caption{Upper bounds on the attacks' success rates as a function of the period length, for $N = 144, K=89, \alpha=0.33$}
		\label{fig:sliding2w-b}
	\end{subfigure}\\[1ex]
	\begin{subfigure}[t]{.5\linewidth}
		\centering
		\includegraphics[scale=0.37]{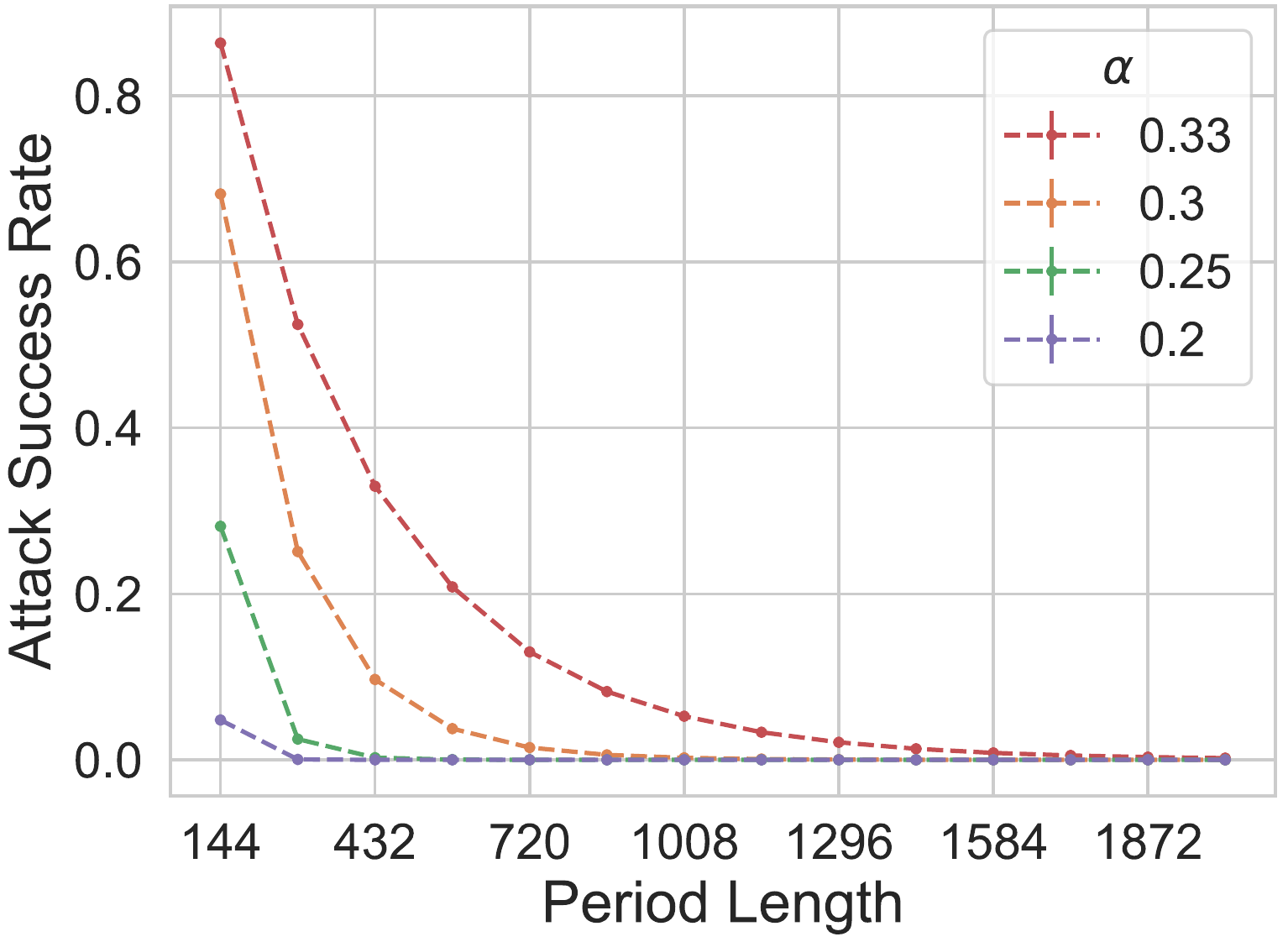}
		\caption{Congestion attack success rate as a function of the period length, for $N = 144, K=89$}
		\label{fig:sliding2w-c}
	\end{subfigure}
	\hspace{1em}%
	\begin{subfigure}[t]{.5\linewidth}
		\centering
		\includegraphics[scale=0.37] {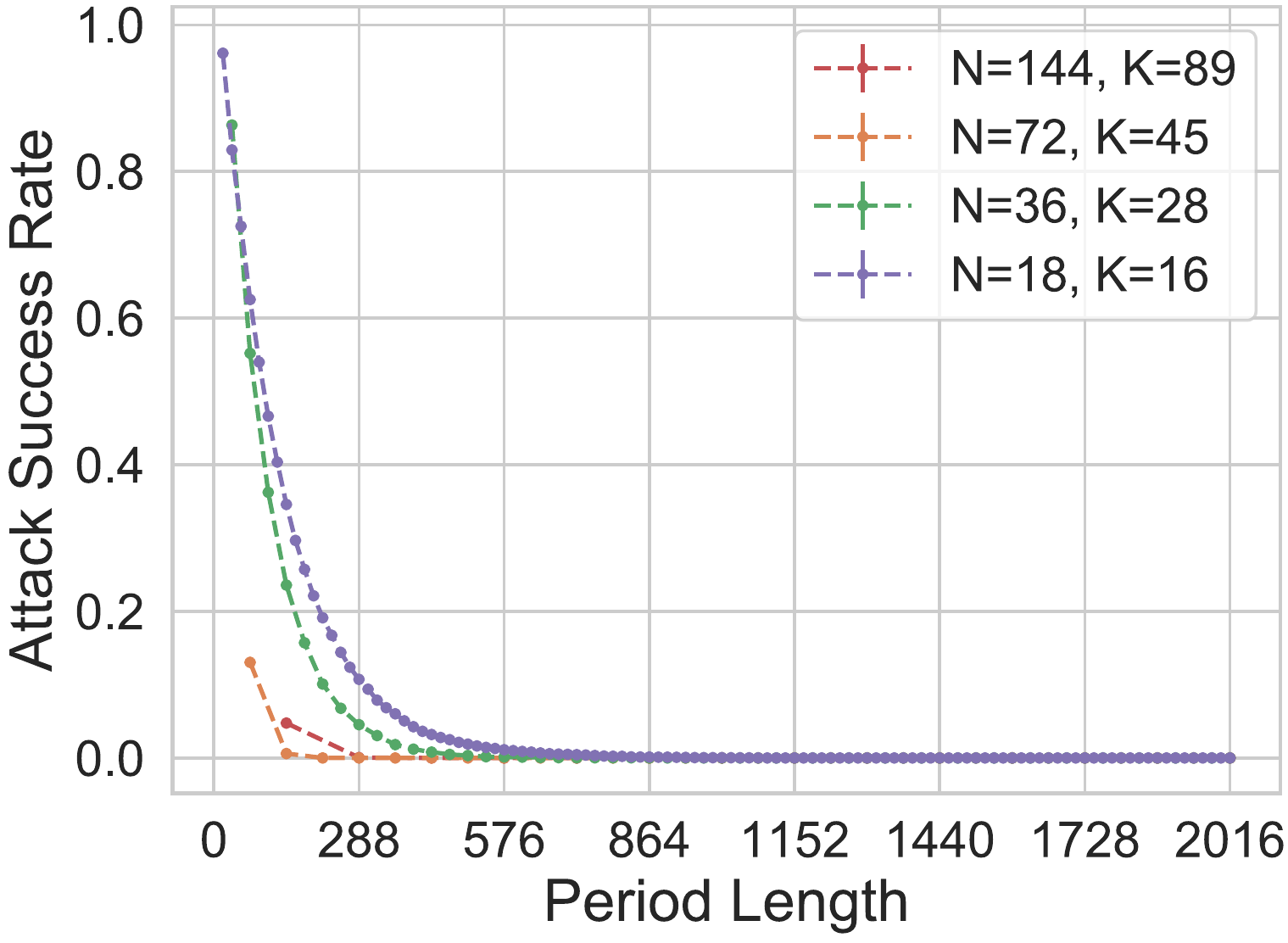}
		\caption{Congestion attack success rate as a function of the period length, for $\alpha=0.2$} 
		\label{fig:sliding2w-d}
	\end{subfigure}
	\caption{Evaluation of the attacks' success rates for $\hat{M}=2016$ (2 weeks in Bitcoin)}
	\label{fig:sliding2w}
\end{figure}

Figure~\ref{fig:sliding2w-a} presents the probability of success in each of the attacks for different $K$ values. As the graph shows, choosing $K=89$ gives protection against both attacks. We compute the upper bounds (from Theorems~\ref{thm:bound_un}-\ref{thm:bound_co}) for this value of $K$ in Figure~\ref{fig:sliding2w-b}. 
The presented bounds as they appear in the graph are loose compared to the simulation results and afford a low level of defense, especially against the congestion attack. These bounds give us useful, but non-tight, upper bounds on the results for periods that are of longer length, for which the probability is extremely small. To get more precise results, we use more simulations to compute the congestion attacks' success rate for different period lengths and present the results in Figure~\ref{fig:sliding2w-c}. Each curve corresponds to a different value of $\alpha$, the computational power of the attacker. The defense rate against congestion attacks is extremely low for short periods. For $\alpha=0.33$, we reach a $>99\%$ defense rate only for periods of $\sim~11$ days or more. Lowering the computational power of the adversary naturally improves these results. For example, considering an attacker with computational power $\alpha=0.2$ results in an above $0.9995$ defense rate for period lengths starting from 2 days.

Finally, in Figure~\ref{fig:sliding2w-d} we consider an attacker with a computational power $\alpha=0.2$ and show the congestion attack success rate for different choices of $N, K$ correspondig to sliding windows of lengths 24/12/6/3 hours. We do not present the uncongestion attack results which had above $99\%$ defense rate for any $N\leq n \leq \hat{M}=2016$. 

We conclude that the longer the periods are, the higher and more effective the protection against attacks is. In Ethereum, we obtained very high defense rates even when choosing short sliding window sizes and against strong attackers. In Bitcoin, on the other hand, we need to compromise on the window size and on the attackers' power to achieve higher defense.

We defined uncongested period protocols and suggested a concrete one, the Sliding Window protocol, which meets our requirements (as defined in Section~\ref{desirableProperties}). In the next section, we will describe how to use an uncongested period protocol to adjust the challenge-response protocol to deal with congested periods.

\subsection{Application to Challenge-Response Protocols}
A challenge-response protocol consists of a challenge that takes effect at time \textbf{$T_c$} and a response deadline \textbf{$T_{rd}$} (see section~\ref{background:challenge-response}).
We link $T_c$ and $T_{rd}$ to their corresponding block height and denote by $b(\textbf{$T$})$ the block at height $T$.  

The parties involved in the challenge decide in advance on an uncongestion period protocol $\IsPeriodUnCongested{}$ to use. We recall that $\IsPeriodUnCongested:\{0,1\}^*\rightarrow \{0,1\}$ accepts a congestion vector (a binary series representing the congestion signal of blocks in a period) and returns 1 if the period is congested and 0 otherwise. 
To apply the uncongestion period protocol, the parties adjust \textbf{$T_{rd}$} to a short deadline that gives them a reasonable time to respond to the challenge assuming an optimal case with no congestion.

The response deadline \textbf{$T_{rd}$} is applied only in the event that the challenge window $Pe=(b(\textbf{$T_c$}),b(\textbf{$T_c+1$}),...,b(\textbf{$T_{rd}$}))$ is uncongested.
In the case where the challenge window is congested, we repeatedly extend \textbf{$T_{rd}$}, 1 block at a time, as long as it remains congested.
To avoid an edge case where the deadline is extended indefinitely, we define \textbf{$\hat{T}_{rd} = T_c + \hat{M}$},
an upper bound on the deadline (see Definition~\ref{def:M}).
The challenge-response protocol adjustment is summarized in the algorithm below. 

\begin{algorithm}[ht]
	\DontPrintSemicolon
	$T_c\gets init$\;
	$T_{rd}\gets init$\;
	$Pe=(b(\textbf{$T_c$}),b(\textbf{$T_c+1$}),...,b(\textbf{$T_{rd}$}))$\;
	$Pe^c\gets congestion\_vector(Pe)$\;
	\While{$\IsPeriodUnCongested(Pe^c)=0$ and $T_{rd}<\textbf{$\hat{T}_{rd}$}$}{
		$T_{rd}\gets T_{rd}+1$\;
		$Pe=(b(\textbf{$T_c$}),b(\textbf{$T_c+1$}),...,b(\textbf{$T_{rd}$}))$\;
	}
\end{algorithm}

We emphasize that the extension of the deadline is not necessarily carried out at the exact moment of the deadline (since smart contract actions need to be triggered by a transaction to the contract). Instead, a transaction that is submitted afterwards is determined to be either before or after the deadline given any possible extensions that are due. The uncongestion period protocol is specified in advance in the smart contract, and the deadline calculation is triggered either by a late response to the challenge or by the challenger that claims that a response did not arrive in time.

\section{Implementation}
\label{sec:Implementation}
We provide an implementation of the Sliding Window protocol as an Ethereum smart contract using the EIP 1559 \emph{base fee} to determine block congestion.
EIP 1559 implements a \emph{base fee} that is adjusted up and down by the protocol according to how congested the network is. 
The EVM supports fetching the \emph{base fee} of the highest (current) block. 
We suggest extending this to fetch the \emph{base fee} of any block, and to add an opcode that checks whether a block is congested (without such opcodes, it is not possible to fully implement the mechanisms put forward in this paper).
This opcode will receive as inputs a block and a maximum base fee (chosen by a user) and will return whether the maximum base fee exceeds the block's base fee.

In the implementation, we set the sliding window size equal to the initial deadline of the examined period (before being granted any extension).

The full \href{https://github.com/stonecoldpat/slidingwindow}{github}\footnote{\href{https://github.com/stonecoldpat/slidingwindow}{https://github.com/stonecoldpat/slidingwindow}} repository includes the smart contracts, the new opcode, and the tests. In addition, we include the Solidity code of the contracts in Appendix~\ref{app:code}.

\section{Conclusion}
\label{sec:Conclusion}
In this paper, we tackled a problem that arises when challenge-response protocols face congested periods.
When the network experiences congestion, 
users will often miss the response deadline, which can lead to serious issues including financial loss. We formalized the problem and proposed a new protocol called the Sliding Window as a solution. Our protocol defines a reliable way to detect congested periods by looking only at the data available on-chain.
We then used this to extend the challenge-response deadline when congestion occurs.
We studied the security of the protocol for different parameters.
Our results showed that it is possible to decrease the 
time settlement (deadline) of challenge-response protocols significantly, while expanding the security of the protocol to deal with cases of congestion.

For future work, it would be interesting to evaluate and optimize this protocol and its security analysis for more realistic congestion settings---in particular, settings in which congestion is correlated between consecutive blocks---and to provide more experimental analysis of these settings. Is is also of interest to explore whether Ethereum's proposed base fee can be used as a sufficiently robust congestion signal. 

\section{Acknowledgments}
	Ayelet Lotem and Aviv Zohar are partially supported by grants from the Israel Science Foundation (grants 1504/17 \& 1443/21) and by a grant from the HUJI Cyber Security Research Center in conjunction with the Israel National Cyber Bureau.

\bibliographystyle{splncs04}
\bibliography{main}

\newpage
\appendix

\ifsubmission{
	\section{Proofs}\label{app:proofs}
	
	\subsection{Proofs of Section~\ref{sec:definitions}}\label{app:proofs1}
	Proof of Proposition~\ref{prop:one}.
	\begin{newProposition}{1}
		A miner manipulating a block $\block$ to make it $(\theta_1,\gamma_1)$-congested when it is not will lose a potential profit of at least  $\blocksizelimit\cdot\int_{1-(\gamma_1-\gamma_{\block}(\theta_1))}^{1} \theta_{\block}(\gamma) \,d\gamma$.
	\end{newProposition}
	\proofpropositionone
	
	Proof of Proposition~\ref{prop:two}.
	\begin{newProposition}{2}
		A miner manipulating a block $\block$ to reverse its signal from $(\theta_1,\gamma_1)$-congested to not congested will lose a potential profit of at least  $\blocksizelimit\cdot\int_{\gamma_1}^{\gamma_{\block}(\theta_1)}(\theta_{\block}(\gamma)-\theta_1) \,d\gamma$.
	\end{newProposition}
	\proofpropositiontwo
	
	\subsection{Proof of Section~\ref{sec:UncongestedPeriodProtocols}}\label{app:proofs2}
	
	\proofconscuivemonotone
	\LconsecEffInfo
	
% 	\begin{newTh}{2}
% 		The probability of an attacker with a relative computational power $\alpha$ to successfully manipulate $Pe$ into a congested period, in a p-congested network
% 		equals    $1-e_1 \cdot \hat{T^n} \cdot e_{L+1}^t$.
% 	\end{newTh}
% 	\prooflconsecutivecongestion
	
% 	\slidingwindowmonotonicity

}
\fi

\section{Block Congestion Definition: Examples}\label{app:block_congestion_def_examples} 
In Definition~\ref{block_congestion_def}, we offered to define the congestion of a block based on the fee densities.
Other, less effective, ways in which block congestion could be defined are listed below:
\begin{itemize}
	\item \textit{By transaction with lowest fee density: }
	A block is $\theta$-congested if the lowest fee density for a transaction in it is bigger than $\theta$: $\min_{tx\in \block}\fee{tx} \geq \theta$.
	However, a miner could decide to always enter a single transaction with a very low fee density (less than $\theta$) and easily change a block from congested to not congested; hence this scheme is easily manipulable.
	\item \textit{By transaction with highest fee density:}
	A block is $\theta$-congested if the highest fee density of its transactions is bigger than $\theta$: $\max_{tx\in \block}\fee{tx} \geq \theta$.
	A miner could decide to always enter a single (dummy) transaction with a high fee density (higher than $\theta$) and easily change a block from not congested to congested; hence this scheme is also easily manipulable.
	\item \textit{By non-zero-fee transaction occupancy ("block size"):}
	A block is $\gamma$-congested if it is full at $\gamma$-fraction of its occupancy with non-zero-fee transactions: $ \sum_{tx\in B:~\fee{tx}>0} \size{tx} \geq \gamma\cdot\blocksizelimit$.
	A miner could artificially add transactions with positive fee density to fill the block at no cost.
	\item \textit{By transaction fees (instead of fee density): }
	A block is ($f$,$\gamma$)-congested if at least a fraction $\gamma$ of it is filled with transactions with fees above some value $f$:
	$\sum_{tx\in \block: \fee{tx}\cdot\size{tx}\geq f}\size{tx} \geq \gamma \cdot \blocksizelimit$.
	A miner could prioritize transactions by their size in order to decide the congestion signal, without lowering his profit from the fees ($\utility{\block}$).
\end{itemize}

\section{Code} \label{app:code}
% \vspace{-75mm}

\begin{figure}[ht]
    \begin{tabular}{l}
		\quad $\contract$ \texttt{SlidingWindow} \is{} \texttt{BlockchainMock} \{ \\ \\
		\quad \quad $\oninput$ \texttt{isPeriodCongested}(\uint{} \texttt{startBlock}, \uint{} \texttt{k}, \uint{} \texttt{n}, \\ 
		\quad \quad \quad \uint{} \texttt{maximumBaseFee}) \public{} \view{} \returns{} (\bool) $\{$\\ \\
		\quad \quad \quad \require(\texttt{n$>=$k}, \requirecomment{`N should be greater than or equal to K.'});\\
		\quad \quad \quad \require(\texttt{blocks.length$>=$n}, \requirecomment{`Total should be greater than or equal to N.'}); \\ \\
		\quad \quad \quad \uint{} \texttt{totalCongested = 0;} \\
		\quad \quad \quad \bool[] \memory{} \texttt{recordCongestion} = \new{} \bool$[]$\texttt{(blocks.length $-$ startBlock);} \\
		\quad \quad \quad \\
		\quad \quad \quad \for(\uint{} \texttt{i=startBlock; i $<$ startBlock+n; i++)} $\{$\\
		\quad \quad \quad \quad \comment{// Keep a record of this block's congestion.}\\ 
		\quad \quad \quad \quad \texttt{recordCongestion[i] = isCongested(i, maximumBaseFee);} \\ 
		\quad \quad \quad \quad \\ 
	    \quad \quad \quad \quad \comment{// Sum of congestion (so far). } \\ 
		\quad \quad \quad \quad \ifs\texttt{(recordCongestion[i])} $\{$ \\ 
		\quad \quad \quad \quad \quad \texttt{totalCongested = totalCongested + 1;}\\ 
		\quad \quad \quad \quad  $\}$\\ 
		\quad \quad \quad \quad  \\ 
		\quad \quad \quad \quad  \ifs\texttt{(totalCongested$>$=k)} $\{$\\ 
		\quad \quad \quad \quad \quad \return{} \trues ; \\
		\quad \quad \quad \quad $\}$ \\ 
		\quad \quad \quad $\}$ \\ 
		\quad \quad \quad \\
		\quad \quad \quad \comment{// Activate the sliding window.}\\
		\quad \quad \quad \for(\uint{} \texttt{i=startBlock+n; i$<$blocks.length; i++)} \\
		\quad \quad \quad \\
		\quad \quad \quad \quad \comment{// Remove start of the window.} \\
		\quad \quad \quad \quad \ifs\texttt{(recordCongestion[i-n]} $\{$ \\
		\quad \quad \quad \quad \quad \texttt{totalCongested = totalCongested - 1;}\\
		\quad \quad \quad \quad $\}$ \\
		\quad \quad \quad \quad \\
		\quad \quad \quad \quad \comment{Keep a record of this block's congestion.} \\
		\quad \quad \quad \quad \texttt{recordCongestion[i] = isCongested(i, maximumBaseFee);} \\
		\quad \quad \quad \quad \\
		\quad \quad \quad \quad \comment{// Add to the end of the window.}\\
		\quad \quad \quad \quad \ifs\texttt{(recordCongestion[i])} $\{$ \\
		\quad \quad \quad \quad \quad \texttt{totalCongested = totalCongested + 1;}\\
		\quad \quad \quad \quad $\}$ \\
		\quad \quad \quad \quad \\
		\quad \quad \quad \quad \ifs\texttt{(totalCongested$>$=k)} $\{$ \\
		\quad \quad \quad \quad \quad \return{} \trues ; \\
		\quad \quad \quad \quad $\}$ \\
		\quad \quad \comment{// Not congested.} \\ 
	    \quad \quad \return{} \falses ; \\ 
		\quad \quad  $\}$ \\
		\quad $\}$
	\end{tabular} 
		
\end{figure}

\begin{figure}[H]
    \begin{tabular}{l}
		\quad $\contract$ \texttt{BlockchainMock} \{ \\ \\
		
		\quad \quad \comment{// Simulate block.basefee(). EVM only fetches the current basefee.} \\
		\quad \quad \struct{} \texttt{Block} $\{$ \uint{} \texttt{baseFee;} $\}$ \\
		
		\quad \quad \texttt{Block[]} \public{} \texttt{blocks;} \\ \\ 
		
		\quad \quad \comment{// Should the caller consider this block congested?} \\
		\quad \quad $\oninput$ \texttt{isCongested}(\uint{} \texttt{blockNumber}, \uint{} \texttt{maximumBaseFee)} \\ 
		
		\quad \quad \public{} \view{} \returns{} (\bool) $\{$ \\ \\
		
        \quad \quad \quad \ifs(\texttt{blocks[blockNumber].baseFee $>$ maximumBaseFee)} $\{$ \\
        \quad \quad \quad \quad \return{} \trues; \\
        \quad \quad \quad $\}$ \\ \\
         
        \quad \quad \quad \return{} \falses; \\
		\quad \quad  $\}$ \\ 
		\quad $\}$
	\end{tabular} 
		
\end{figure}

\begin{figure}[H]
    \begin{tabular}{l}
		\quad $\contract$ \texttt{Auction} \is{} \texttt{SlidingWindow} \{ \\ \\
		
		\quad \quad \bool{} \texttt{start = \falses;} \\
		\quad \quad \uint{} \texttt{startBlock;} \\
		\quad \quad \uint{} \texttt{k;} \\
		\quad \quad \uint{} \texttt{n;} \\ 
		\quad \quad \uint{} \texttt{gasPriceCeiling;} \\ \\
		\quad \quad \oninput{} \texttt{startAuction}(\uint{}  \texttt{\_startBlock}, \uint{} 
		\texttt{\_k}, \uint{} \texttt{\_n},  \\
		\quad \quad \uint{} \texttt{\_gasPriceCeiling}) \public{} $\{$ \\
		\quad \quad \quad \texttt{startBlock = \_startBlock;} \\
		\quad \quad \quad \texttt{k = \_k;} \\
		\quad \quad \quad \texttt{n = \_n;} \\
		\quad \quad \quad \texttt{gasPriceCeiling = \_gasPriceCeiling;} \\
		\quad \quad \quad \texttt{start =} \trues; \comment{// Kick-start the auction!} \\
		\quad \quad $\}$ \\ \\
		
		\quad \quad \oninput{} \texttt{finaliseAuction()} \public{} \returns(\bool) $\{$ \\
		\quad \quad \quad \ifs(\texttt{isPeriodCongested(startBlock, k, n, gasPriceCeiling))} $\{$ \\
		\quad \quad \quad \quad \return{} \falses; \\
		\quad \quad \quad $\}$ \\ \\

        \quad \quad \quad \return{} \trues; \\
		\quad \quad  $\}$ \\
		\quad $\}$
	\end{tabular} 
		
\end{figure}

	\clearpage
	\addcontentsline{toc}{section}{Appendices}
	\renewcommand{\thesubsection}{\Alph{subsection}}

\end{document}

%% file: macros.tex
\newcommand{\floor}[1]{\lfloor #1 \rfloor}
\newcommand{\ayelet}[1]{{\color{blue}\textbf{Ayelet}: #1}}
\newcommand{\sarah}[1]{{\color{purple}\textbf{Sarah}: #1}}
\newcommand{\fixes}[1]{{\color{red} #1}}
\newcommand{\patrick}[1]{{\color{red}\textbf{Patrick}: #1}}
\newcommand{\aviv}[1]{{\color{green}\textbf{Aviv}: #1}}
\newcommand{\todo}[1]{{\color{orange}\textbf{TODO}: #1}}
\newcommand{\change}[1]{{\color{red} #1}}
\newcommand{\block}{\mathbf{B}}
\newcommand{\maxfee}{\mathsf{fee_{max}}}
\newcommand{\blockcongestion}{\textsf{BlockCongested}}
\newcommand{\txset}{\mathcal{T}}
\newcommand{\tx}{\mathsf{tx}}

\newcommand{\blocksizelimit}{\ensuremath{\mathcal{B}}}
\newcommand{\fee}[1]{\ensuremath{\phi(#1)}}
\newcommand{\size}[1]{\ensuremath{w(#1)}}
\newcommand{\weight}[2]{\ensuremath{\mathcal{W_{#2}(#1)}}}
\newcommand{\utility}[1]{\ensuremath{\mathcal{U_{#1}}}}
\newcommand{\IsPeriodUnCongested}[1]{\textsf{UCP{#1}}}
\newcommand{\UCP}[2]{\textsf{UCP}_{#1}(#2)}
\newcommand{\IsBlockCongested}[2]{\mathcal{C_{#1}(#2)}}

\newenvironment{newProposition}[1]{%
  \renewcommand\theproposition{#1}
  \proposition
}{\endtheorem}

\newenvironment{newTh}[1]{%
  \renewcommand\thetheorem{#1}
  \theorem
}{\endtheorem}